\documentclass[reqno]{amsart}
\usepackage{amsmath}
\usepackage{amssymb}
\usepackage{amsthm}
\usepackage{amsfonts}
\usepackage{amstext}
\usepackage{amsopn}
\usepackage{amsxtra}
\usepackage{pst-all}
\usepackage[latin1]{inputenc}
\usepackage{mathrsfs}
\usepackage{dsfont}
\usepackage{hyperref}

\newcommand\R{{\ensuremath {\mathbb R}}}
\newcommand\C{{\ensuremath {\mathbb C}}}
\newcommand\N{{\ensuremath {\mathbb N}}}

\renewcommand\epsilon{\varepsilon}
\renewcommand\phi{\varphi}
\newcommand{\gH}{\mathfrak{H}}

\newcommand{\cB}{\mathcal{B}}
\newcommand{\cL}{\mathcal{L}}

\newcommand{\cE}{\mathcal{E}}

\newcommand\ii{{\ensuremath {\infty}}}
\newcommand\pscal[1]{{\ensuremath{\left\langle #1 \right\rangle}}}
\newcommand{\norm}[1]{ \left| \! \left| #1 \right| \! \right| }

\newtheorem{thm}{Theorem}
\newtheorem{lemma}{Lemma}

\theoremstyle{definition}

\newtheorem{remark}{Remark}

\newcommand\vp{{\vec p}}
\newcommand\vx{{\vec x}}
\newcommand\vy{{\vec y}}
\newcommand\vL{{\vec L}}
\newcommand\ve{{\vec e}}
\newcommand\vA{{\vec A}}
\newcommand\vB{{\vec B}}

\title[Strongly correlated phases in rapidly rotating Bose gases]{Strongly 
correlated phases in rapidly rotating Bose gases}

\thanks{\copyright\, 2009 by
    the authors. This paper may be reproduced, in its entirety, for
    non-commercial purposes.}

\author[M. Lewin]{Mathieu LEWIN}
 \address{CNRS \& Laboratoire de Mathématiques UMR 8088, Université de Cergy-Pontoise, 2 Avenue Adolphe Chauvin, 95302 Cergy-Pontoise Cedex, France.}
  \email{Mathieu.Lewin@math.cnrs.fr}

\author[R. Seiringer]{Robert SEIRINGER}
 \address{Department of Physics, Princeton University, Princeton NJ 08544, USA.}
 \email{rseiring@princeton.edu}

\begin{document}

\date{June 3, 2009}

\begin{abstract}
  We consider a system of trapped spinless bosons interacting with a
  repulsive potential and subject to rotation.  In the limit of rapid
  rotation and small scattering length, we rigorously show that the
  ground state energy converges to that of a simplified model
  Hamiltonian with contact interaction projected onto the Lowest
  Landau Level. This effective Hamiltonian models the bosonic analogue
  of the Fractional Quantum Hall Effect (FQHE). For a fixed number of
  particles, we also prove convergence of states; in particular, in a
  certain regime we show convergence towards the bosonic Laughlin
  wavefunction. This is the first rigorous justification of the
  effective FQHE Hamiltonian for rapidly rotating Bose gases. We
  review previous results on this effective Hamiltonian and outline
  open problems.
\end{abstract}

\maketitle

\section{Introduction}

A fundamental characteristic of trapped Bose gases is
their response to rotation \cite{DalGioPitStr-99,Cooper-08}. When the
angular velocity $\Omega$ becomes large, a transition from a condensed
regime to a highly correlated, uncondensed, phase is expected. The 
behavior of the system then has certain features similar to the Fractional Quantum
Hall Effect that is observed in superconductors submitted to a
magnetic field. For rotating Bose gases, this regime has not yet been
observed experimentally, the corresponding value of $\Omega$ being
unattainable at present. Nevertheless, there has been a lot of
interest in the theoretical understanding of this phenomenon in the
literature (see, e.g., \cite{Cooper-08} for a recent review).

In this paper we present a rigorous study of such a Bose system with a
generic repulsive interaction potential. We show that in a certain
limit, the Hamiltonian of the system can be replaced by a simplified
effective Hamiltonian with a contact interaction, projected onto the
Lowest Landau Level (LLL). This simplified model has been extensively used in the physics literature
\cite{DalGioPitStr-99,CooWil-99,CooWilGun-01,PapBer-01,MorFed-06,RegJol-07,MasMatOuv-07}. We
also prove that in the same limit, the true ground state of the system
converges to the ground state of the model Hamiltonian. In a certain
parameter regime, our analysis provides a rigorous derivation of the
bosonic equivalent of the well-known Laughlin state
\cite{Laughlin-83}.

Let us consider $N$ interacting spinless bosons submitted to a
rotation around the $x^3$ axis and a harmonic trapping
potential. Denoting $\vx = (x^1,x^2,x^3)\in\R^3$, the Hamiltonian of the
system in the rotating frame is given by
\begin{equation}
\sum_{j=1}^N\left[\frac{|\vp_j|^2}{2m}+\frac{m}{2}\left(\omega_{\perp}^2\left(|x_j^1|^2+|x_j^2|^2\right)+\omega_{\parallel}^2|x_j^3|^2\right)-\Omega \ve_3\cdot \vL_j\right] + \sum_{1\leq j< k\leq N}W_{a}(\vx_j-\vx_k).
\label{def:Hamil_phys_const}
\end{equation}
Here $\vL=\vx\times \vp$ is the angular momentum, $\ve_3=(0,0,1)$, $m$
is the mass of the bosons, and $\omega_\perp$ and $\omega_\parallel$
are the trap frequencies. The interaction potential $W_{a}$ is assumed
to be non-negative, i.e., purely repulsive, and to have scattering
length $a$ (see \cite{LieSeiSolYng-05} for a proper definition of the
scattering length). It is natural to introduce a fixed potential $W$
with scattering length $1$ and write $W_{a}(x)=a^{-2}W(x/ a)$. For
convenience we will assume that the angular velocity $\Omega$ is
nonnegative. The above Hamiltonian acts on the space of
permutation-symmetric square-integrable $N$-body wavefunctions.

The Hamiltonian~\eqref{def:Hamil_phys_const} is stable
(i.e., bounded from below) only when $\Omega\leq\omega_\perp$. The
regime of rapid rotation that will be of special interest to us
corresponds to the case when $\Omega$ is very close to the
maximal possible speed $\omega_\perp$, i.e., 
$$
\omega := \frac {\omega_\perp-\Omega}{\omega_\perp} \ll 1 \,.
$$

To simplify certain expressions, we will work with an isotropic
harmonic potential, $\omega_\perp=\omega_\parallel$. Our results would
hold equally well when $\omega_\perp\neq\omega_\parallel$ but
$\omega_\parallel/\omega_\perp\geq \epsilon>0$. Similarly, a
non-harmonic confinement potential in the $x^3$ direction could be
used.
 
It is convenient to chose units such that
$m=\hbar=\omega_\perp=1$. Introducing $\vA(\vx)=(-x^2,x^1,0)$ and
completing the square, our Hamiltonian \eqref{def:Hamil_phys_const} can be written as
\begin{equation}
\boxed{H^N_{\omega,a}:= \sum_{j=1}^N\left[\frac{|\vp_j-\vA(\vx_j)|^2+|x_j^3|^2-3}2+\omega \ve_3\cdot \vL_j\right] + \sum_{1\leq j< k\leq N}W_a(\vx_j-\vx_k).}
\label{def:Hamil}
\end{equation} 
The kinetic energy term of this Hamiltonian is equivalent to that of a
charged particle in a constant magnetic field $\vB=\vec\nabla\times \vA$. The
spectrum of $(|\vp-\vA(\vx)|^2+|x^3|^2)/2$ is purely discrete, its
eigenvalues being $3(j+1/2)$ for $j=0,1,\dots$. They are all
infinitely degenerate. In the definition \eqref{def:Hamil} of our
Hamiltonian we have subtracted the unimportant ground state energy
$3/2$ of the kinetic term.

When the speed of rotation $\Omega$ is not too close to $\omega_\perp$
and the Bose gas is sufficiently dilute, the ground state of
(\ref{def:Hamil}) is known to exhibit Bose-Einstein condensation, with
condensate wavefunction described by ground states of the
Gross-Pitaevskii functional \cite{Gross-61,Pitaevskii-61}
\begin{equation}
  \cE^{\rm GP}(\phi)=\pscal{\phi,\left(\frac{|\vp-\vA(\vx)|^2+|x^3|^2-3}2+\omega \ve_3\cdot \vL\right)\phi}+\frac{g}2\int_{\R^3}|\phi|^4,
\label{def:Gross_Pit} 
\end{equation}
where $g=4\pi Na$. In the limit $N\to \infty$ with $0<\omega\leq 1$ and $g>0$
fixed this was proved in \cite{LieSei-02,LieSei-06}.

The properties of the Gross-Pitaevskii ground state in the rapidly
rotating regime $\omega\to0$ were intensely studied in the literature,
both from a numerical \cite{ButRok-99,CooKomRea-04,AftBlaDal-05} and
an analytical \cite{AftBlaNie-06,AftBla-06} point of view.  As the
speed of rotation increases, more and more vortices appear and the
wavefunction acquires a higher angular momentum. The location of these
vortices is conveniently studied in the \emph{Lowest Landau Level}
(LLL) approximation where one restricts $\phi$ to lie in the kernel of
$|\vp-\vA(\vx)|^2+|x^3|^2-3$.  This LLL approximation is justified
\cite{AftBla-08} when $\omega\ll 1$ and $g\omega\ll 1$, the number of
vortices being then proportional to
$$N_v\sim \sqrt{\frac{g}{\omega}}\sim \sqrt{\frac{Na}{\omega}}.$$

The Gross-Pitaevskii functional (\ref{def:Gross_Pit}) is expected to
be an accurate description of the ground state of the many-body system
(\ref{def:Hamil}) provided the number of vortices is much smaller than
the number of particles in the system, i.e., when
\begin{equation}
 \frac{a}{N\omega}\ll1.
\label{cond_condensation}
\end{equation}
Within the LLL approximation, this was recently shown in
\cite{LieSeiYng-09} to be indeed the case. In terms of the
\emph{filling factor} $\nu=N^2/(2L_{\rm tot})$ \cite{CooWilGun-01},
\eqref{cond_condensation} corresponds to $\nu\gg1$.

When $a/(N\omega)$ is not small, a completely different regime is
expected. Evidence of strongly correlated states was found in exact
diagonalization studies of small systems. As the rotation frequency
$\Omega$ is increased from 0 to its upper limit $\omega_\perp$, the
ground state encounters a series of transitions between certain values
of the angular momentum. The behavior of the system is similar to the
\emph{Fractional Quantum Hall Effect} in fermionic systems, and
usually modeled by an effective Hamiltonian with contact interaction
in the LLL. The rigorous derivation of this effective Hamiltonian for
Bose gases with generic repulsive two-body interactions is the main
purpose of this paper.

\section{Main Results}

\subsection{Derivation of the Effective Hamiltonian on the LLL}

The ground state energy in the bosonic sector is given by
\begin{equation}
  \boxed{E_N(\omega,a):=\inf\sigma_{\bigvee_1^NL^2(\R^3,\C)}^{\phantom{\sigma_{\bigvee_1^NL^2(\R^3,\C)}}}\left(H^N_{\omega,a}\right)}
\label{def_GS_energy}
\end{equation}
where $\bigvee_1^NL^2(\R^3,\C)$ denotes the \emph{symmetric} tensor product and $H^N_{\omega,a}$ was defined above in \eqref{def:Hamil}. 

We will compare the ground state energy of the above Hamiltonian
\eqref{def:Hamil} with the simplified model consisting in restricting
the wavefunction to the \emph{$N$-body Lowest Landau Level} (LLL) and
replacing $W_a$ by $4\pi a$ times a contact interaction potential.
The LLL is defined as the ground state eigenspace of the kinetic part
of the operator (\ref{def:Hamil}) at $\omega=0$.  This subspace of
$\bigvee_1^NL^2(\R^3,\C)$ reads
\begin{multline}
\gH_N:=\bigg\{ \Psi(\vx_1,\dots,\vx_N)=F(x^1_1+ix_1^2,\dots,x^1_N+ix_N^2)e^{-\sum_{i=1}^N\frac{|\vx_i|^2}2}\in L^2(\R^{3N})\ :\\ (z_1,\dots,z_N)\mapsto F(z_1,\dots,z_N)\text{ is holomorphic and symmetric}\bigg\} 
\end{multline}
where we denote $\vx=(x^1,x^2,x^3)\in\R^3$ as before. We will use use
the notation $z=x^1+ix^2\in\C$ and we will sometimes identify it with
$(x^1,x^2)\in\R^2$. For any $\Psi\in\gH_N$, we have by definition
$$
\sum_{j=1}^N\left[\frac{|\vp_j-\vA(\vx_j)|^2+|x_j^3|^2-3}2\right]\Psi=0\,.
$$
The ground state energy of the simplified effective model in the LLL is given by
\begin{equation}
\boxed{E^{\rm LLL}_N(\omega,a):=\inf_{\substack{\Psi\in\gH_N\\ \norm{\Psi}=1}}\pscal{\Psi,\left(\omega\sum_{j=1}^N \ve_3\cdot \vL_j+4\pi a \sum_{1\leq i<j\leq N}\delta(\vx_i-\vx_j)\right)\Psi}.}
\label{def:GS_energy_LLL}
\end{equation}
Note that although it makes no sense to use a delta potential in the original
Hilbert space, functions in the space $\gH_N$ are all smooth, hence
$$
\pscal{\Psi,\delta(\vx_1-\vx_2)\Psi}=\int\cdots\int |\Psi(\vx_2,\vx_2,\vx_3,\dots,\vx_N)|^2dx_2\cdots dx_N
$$
makes perfect sense and defines a bounded selfadjoint operator. As
we will discuss in the next section, for any $\omega\geq 0$ and $a\geq 0$ there exists a
ground state $\Psi\in\gH_N$ for the problem
\eqref{def:GS_energy_LLL}. 

We emphasize that (\ref{def:GS_energy_LLL}) is \emph{not} obtained by
restricting $H^N_{\omega,a}$ to the LLL. For small scattering length $a$ such a
restriction would lead to a similar expression but with the wrong
prefactor $\int W_a$ instead of $4\pi a$ in front of the
$\delta$-interaction (see Remark \ref{rmk:lower_bound_lemma} below). In order to obtain the scattering length, it is
important to note that the LLL restriction is unphysical on length
scales much smaller than the effective \lq\lq magnetic length\rq\rq,
which is $1$ in our units. If $a\ll 1$, the scattering process is
unaffected by the rotation of the system and hence leads to the scattering length as
an effective coupling constant.

Our main result is the following.

\begin{thm}[{\bf Validity of effective LLL model}]\label{thm:bounds}
Let $W$ be a nonnegative radial function such that $\int_{|\vx|>R} W(\vx)\,dx<\ii$ for some $R>0$, with scattering length 1. We define $W_a:=a^{-2}W(\,\cdot\, /a)$ and
\begin{equation}
\kappa:=\frac{a}{N\omega}.
\label{def:kappa}
\end{equation}
\medskip

\noindent $(i)$ \textbf{Upper bound.}  Assume that $\eta:=\kappa^{1/4} a N^{1/2}\leq 1$. For $\kappa^{-3/2}a< C^{-1}$ one has 
\begin{multline}
E_N(\omega,a)\leq E^{\rm LLL}_N(\omega,a) \left(1-Ca \kappa^{-3/2} \right)^{-1} \\ \times \left( 1+ \frac{C \eta}{\min\{1,\kappa N^2\}}\left[1+ \frac {\kappa^{-3/4}}{\sqrt N} + \frac 1\eta \int_{|\vx|\geq (2/\eta)^{3/4}} W \right]
\right)
\label{eq:upper_bound}
\end{multline}
for some universal constant $C>0$.

\medskip

\noindent $(ii)$ \textbf{Lower bound.}  Let $r=\min\{1,\kappa a^{-2/3}\}$. Then  
 \begin{equation}
 E_N(\omega,a ) \geq E^{\rm LLL}_N(\omega,a)\left(1- C \left[ \frac{ N a^{1/3}}r + a^{1/9} r \right] - \frac{1}{4\pi}\int_{|\vx|\geq r^{1/6} a^{-8/9}} W \right)
\label{eq:lower_bound}
\end{equation}
for some universal constant $C>0$.
\end{thm}

What Theorem \ref{thm:bounds} says is that when $\kappa$ stays away
from zero (for instance $\kappa>0$ fixed) and when $a$ is small enough
(depending on the particle number $N$), then one can replace the
problem of minimizing $H^N_{\omega,a}$ with a generic interaction of
scattering length $a$ by the study of a simplified Hamiltonian acting
on the LLL, with a contact interaction of strength
$4\pi a$. The latter model has some very specific features that we
will recall in the next section.

For fixed $\kappa$, the leading order correction in our upper bound is
of the order $aN^{1/2}$ as long as $W$ decays at least as
$|\vx|^{-3-4/3}$ at infinity; it is $(aN^{1/2})^{3\epsilon/4}$ if $W$
decays as $|\vx|^{-3-\epsilon}$ for $0<\epsilon< 4/3$ instead.  The
error bounds in the lower bound are significantly worse. For fixed
$\kappa$ the leading error term is $N a^{1/3}$. It remains a
challenging open problem to derive bounds that display a better $N$
dependence.  These will require a better understanding of the FQHE
regime for large $N$.  One would expect that, for $\kappa$ fixed,
there exist error bounds that are independent of $N$.

The proof of Theorem \ref{thm:bounds} uses several previous ideas
\cite{LieSeiSolYng-05}. The upper bound requires a two-scale trial
function, as suggested first by Dyson in \cite{Dyson-57}, in order to
obtain the scattering length from $W_a$. The fact that the ground
state in $\gH_N$ for \eqref{def:GS_energy_LLL} is not very well known
(contrarily to the condensed Gross-Pitaevskii case) is an important
obstacle, however. As usual, the lower bound is the hardest part and
consequently our conditions on $a$ are more restrictive.

The proof of Theorem \ref{thm:bounds} will be given in Section
\ref{sec:proof_Thm}. More general upper and lower bounds on $E_N(\omega,a)$ are stated in
\eqref{eq:last_estimate_upper_bound}--\eqref{eq:last_ub2} and
\eqref{eq:last_lower_bound}, respectively.

\medskip

\subsection{Effective Hamiltonian on the LLL and Convergence of States}

The result of Theorem~\ref{thm:bounds} can be extended to obtain
information not only the ground state energy but also on the
corresponding eigenfunctions. Before we state our result on the
convergence of ground states in the limit $\omega\to0$ and $a\to0$
with $N$ fixed, we recall in this section several important properties
of the effective Hamiltonian on the LLL.

It is convenient to introduce the  \emph{Bargmann space} \cite{Bargmann-62}
\begin{multline*}
\cB_N:=\bigg\{ F:\C^N\to\C^N \text{ holomorphic and symmetric}\ :\\ \int_\C\cdots \int_\C |F(z_1,\dots,z_N)|^2e^{-\sum
_{j=1}^N|z_j|^2}dz_1\cdots dz_N<\ii\bigg\} 
\end{multline*}
endowed with the scalar product 
$$\pscal{F,G}_{\cB_N}:=\int_\C dz_1\cdots\int_\C dz_N\; \overline{F(z_1,\dots,z_N)} G(z_1,\dots,z_N)e^{-\sum_{j=1}^N|z_j|^2},$$
and its associated norm.
It can easily be checked that if $F\in\cB_N$ then the function $\Psi$ defined by 
$\Psi=\pi^{-N/4}e^{-\sum_{j=1}^N|\vx_j|^2/2}F\in\gH_N$ satisfies $\|\Psi\|_{L^2(\R^{3N})} = \|F\|_{\cB_N}$ and 
\begin{equation}
 \sum_{j=1}^N\left[\frac{|\vp_j-\vA(\vx_j)|^2+|x_j^3|^2-3}2+\omega \ve_3\cdot \vL_j\right] \Psi=\omega\;\left(\sum_{j=1}^N z_j\partial_{z_j}F\right)e^{-\sum_{i=1}^N\frac{|\vx_i|^2}2}.
\label{eq:angular_momentum} 
\end{equation}
The delta interaction potential is defined on $\cB_N$ as follows:
\begin{equation}
\left(\sum_{i< j}\delta_{ij}\right)F:=\frac{1}{(2\pi)^{3/2}}\sum_{i< j}F\left(z_1,\dots,\frac{z_i+z_j}{2},\dots, \frac{z_i+z_j}{2},\dots,z_N\right).
\label{def:delta_ij}
\end{equation}
The prefactor has been chosen to ensure that 
$$
\pscal{F,\delta_{12}F}_{\cB_2}=\int_{\R^{3}}|\Psi(\vx,\vx)|^2 dx\,.
$$
It can easily be seen that  $0\leq \delta_{12}\leq (2\pi)^{-3/2}$, hence $\delta_{12}$ is a bounded self-adjoint operator on $\cB_2$.
The model Hamiltonian acting on $\cB_N$ is defined as
\begin{equation}
\boxed{\tilde{H}^N_{\omega,a}:=\omega\sum_{j=1}^Nz_j\partial_{z_j}+4\pi a\sum_{1\leq i<j\leq N}\delta_{ij} }
\label{def:Hamil_LLL}
\end{equation}
and its ground state energy equals the LLL energy
(\ref{def:GS_energy_LLL}) introduced in the previous section:
$$\boxed{E^{\rm LLL}_N(\omega,a):=\inf\sigma_{\cB_N}(\tilde{H}^N_{\omega,a}).}$$

We introduce, for convenience, the notation
$$\cL_N:=\sum_{j=1}^Nz_j\partial_{z_j}\quad\text{and}\quad \Delta_N:=\sum_{1\leq i<j\leq N}\delta_{ij}$$
for the total angular momentum and the contact interaction potential
in the LLL, respectively. Because of rotation invariance of the
interaction these two operators commute on $\cB_N$, i.e.,
$[\cL_N,\Delta_N]=0$.  Hence the ground state energy $E^{\rm
  LLL}_N(\omega,a)$ of our Hamiltonian
$\tilde{H}^N_{\omega,a}=\omega\cL_N+4\pi a\Delta_N$ is obtained
by looking at the \emph{joint spectrum} of $\Delta_N$ and $\cL_N$.  If
we denote by $\Delta_N(L)$ the lowest eigenvalue of the operator
$\Delta_N$ in the sector of total angular momentum $L$, we get
$$E^{\rm LLL}_N(\omega,a)=\inf_{L\in\N}\big\{\omega L+4\pi a\Delta_N(L)\big\}.$$
Multiplying any common eigenstate of $\cL_N$ and $\Delta_N$ by the
center of mass $\sum_{j=1}^Nz_j$, one sees that
$\sigma\big(\Delta_N\big)_{\restriction \ker(\cL_N-L)}\subset
\sigma\big(\Delta_N\big)_{\restriction \ker(\cL_N-L-1)}$. Therefore,
$L\mapsto \Delta_N(L)$ is nonincreasing.

A sketch of the general form of the joint spectrum of $\Delta_N$ and
$\cL_N$ is shown in Figure \ref{fig:joint_spectrum}. The possible
ground states for $\tilde{H}^N_{\omega,a}$ are those whose values of
$\cL_N$ and $\Delta_N$ lie on the so-called \emph{yrast
  curve}\footnote{In the literature the graph of $L\mapsto\Delta_N(L)$
  is sometimes called the yrast curve. We keep this name for the
  convex hull which contains all the possible ground states of
  $\tilde{H}^N_{\omega,a}$.} \cite{Mottelson-99} which is the graph of
the convex hull of $L\mapsto\Delta_N(L)$. We can write
$$\tilde{H}^N_{\omega,a}=4\pi Na\; \left(\frac{1}{4\pi \kappa}\frac{\cL_N}{N^2}+\frac{\Delta_N}{N}\right)$$
where, as before $\kappa=a/(N\omega)$. Thus the ground state which
will be picked by the system only depends  on the value
of $\kappa$. It jumps from one state to another when $\kappa$
is varied. The FQHE regime corresponds to $\kappa\sim 1$ in which case
$\cL_N\sim N^2$ and $\Delta_N\sim N$, hence $E^{\rm
  LLL}_N(\omega,a)\sim Na$.

\begin{figure}[ht]
\small
\input{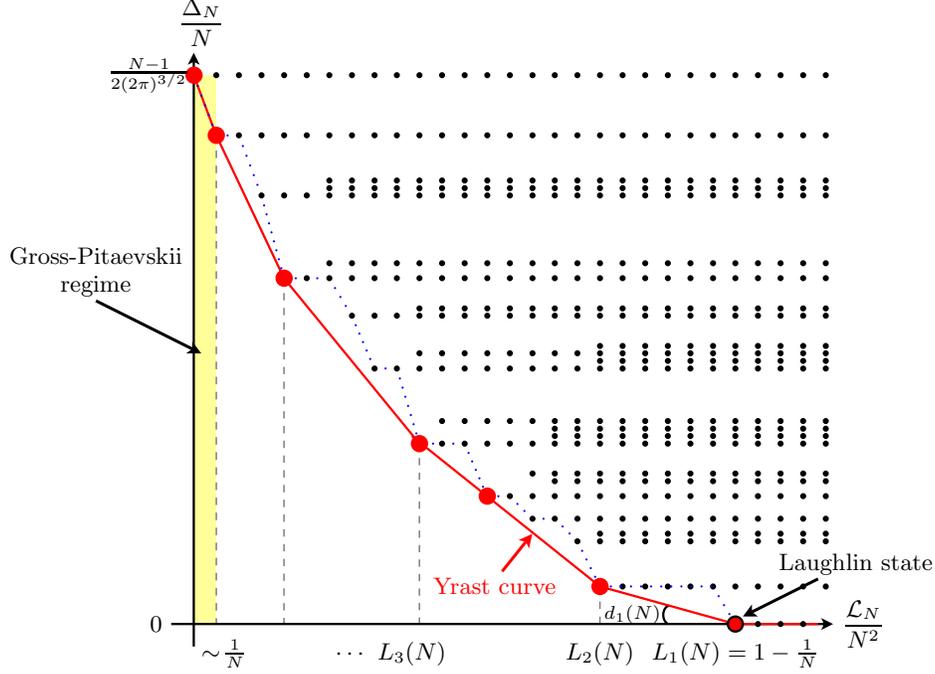}
\caption{\sl General form of the joint spectrum of $\Delta_N/N$ and
  $\cL_N/N^2$. The dashed curve is the graph of
  $\ell\mapsto\Delta_N(\ell N^2)/N$, whereas the solid one is the
  \emph{yrast curve}. Points of the joint spectrum lying on the yrast
  curve are emphasized by thick dots. This figure represents only a sketch, numerical studies for the joint spectrum of $\Delta_N$
  and $\cL_N$ can be found in
  \cite{VieHanRei-00,RegJol-04,BakYanLan-07}. \label{fig:joint_spectrum}}
\end{figure}

\begin{figure}[ht]
\small
\input{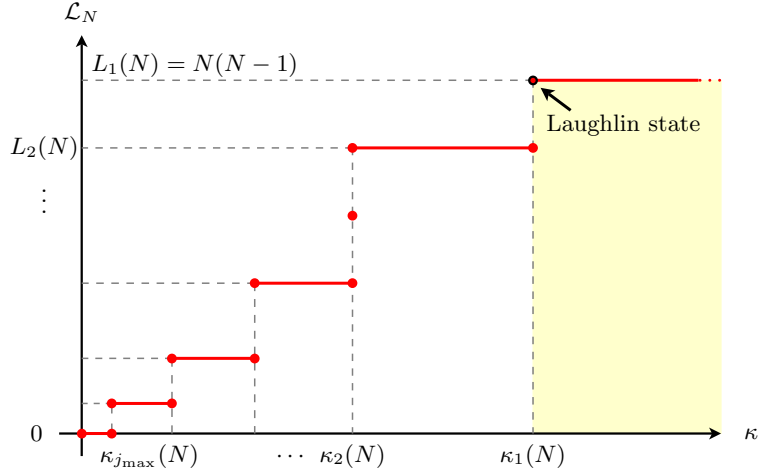}
\caption{\sl Value of the angular momentum of the ground state(s) of
  $\tilde{H}^N_{\omega,a}$, depending on the value of $\kappa$. The
  Laughlin state is the unique ground state for all
  $\kappa>\kappa_1(N)$. The constant function is the unique ground
  state for all $0\leq \kappa<\kappa_{j_{\rm
      max}}(N)$. \label{fig:kappa_L}}
\end{figure}

The kernel of $\Delta_N$ is obviously given by
$$\ker(\Delta_N)=\left\{F(z_1,\dots,z_N)\prod_{1\leq i< j\leq N}(z_i-z_j)^2\ |\ F\text{ holomorphic and symmetric}\right\}.$$
The function which has the lowest angular momentum among these
functions (hence lies on the yrast curve) is the (bosonic)
\emph{Laughlin wavefunction}
$$\boxed{F^N_{\rm Lau}(z_1,\dots,z_N)=k_N\prod_{1\leq i< j\leq N}^{\phantom{N}}(z_i-z_j)^2}$$
where $k_N$ is a normalization factor. It satisfies
$$\cL_NF^N_{\rm Lau}=N(N-1)\,F_{\rm Lau}^N.$$
The Laughlin wavefunction is the unique ground state of
$\tilde{H}^N_{\omega,a}$ as soon as $\kappa>\kappa_1(N):=-1/(4\pi
d_1(N))$, where $d_1(N)$ is the (unknown) left derivative at $\ell=1-1/N$ of the
convex hull of $\ell\mapsto \Delta_N(N^2\ell)/N$. The \emph{filling factor} 
\cite{CooWilGun-01} of the Laughlin function is given by
$$
\nu_{\rm Lau}=\frac{N^2}{2\pscal{ F^N_{\rm Lau},\cL_NF^N_{\rm
      Lau}}}=\frac{1}{2(1-1/N)}\xrightarrow{N\to \ii}\frac12\,.
$$

In general, the convex hull $\tilde{\Delta}_N$ of the function
$L\mapsto \Delta_N(L)$ is piecewise linear and we may define similarly
$\kappa_1(N)>\kappa_2(N)>\cdots >\kappa_k(N)$ by the formula
$\kappa_j(N):=-1/(4\pi d_j(N))$ where $d_1(N)>d_2(N)>\cdots>d_k(N)$
are the successive left derivatives of the function
$\ell\mapsto\tilde\Delta_N(N^2\ell)/N$.  To any $\kappa_j(N)$ we can
associate a unique total angular momentum $L_j(N)$ which is the
highest among states lying on the yrast curve and having a left
derivative equal to $d_j(N)$. The corresponding eigenspace is easily
seen to be the ground state eigenspace of $\tilde{H}^N_{\omega,a}$ when
$\kappa\in(\kappa_{j+1}(N),\kappa_j(N))$.  When $\kappa=\kappa_j(N)$,
the ground state eigenspace is the one containing all states lying on the
yrast line with slope $d_j(N)$. It does not have a unique 
angular momentum. These statements are illustrated in Figure
\ref{fig:kappa_L}.

The only state having $L=0$ is the condensed state
$F(z_1,\dots,z_N)=1$, hence
$\Delta_N(0)=(2\pi)^{-3/2}N(N-1)/2$. Also $\Delta_N(1) = (2\pi)^{-3/2}N(N-1)/2$, with unique state $F(z_1,\dots,z_N) = \sum_{i=1}^N z_i$. Moreover, it is well known
\cite{BerPap-99,PapBer-01,HusVor-02} that
\begin{equation}\label{pap}
\Delta_N(L)=\frac{1}{2 (2\pi)^{3/2}}N\left(N-1-\tfrac 12 L\right) \quad \text{for $2\leq L\leq N$.} 
\end{equation}
For the proof, one notes that $\delta_{12}$ has eigenvalues $0$ and $(2\pi)^{-3/2}$ and commutes with the relative angular momentum $L_{12}  =(z_1-z_2) \left( \partial_{z_1} -\partial_{z_2} \right)/2$; in fact, $\delta_{12}$ is nonzero only on the subspace where $L_{12} = 0$. On {\it 
symmetric} functions of $z_1$ and $z_2$, the smallest non-zero eigenvalue of
$L_{12}$ is 2, hence $(2\pi)^{3/2}\delta_{12} \geq 1 - L_{12}/2$. Summing over all pairs
we get 
$$ 
(2\pi)^{3/2} \sum_{i<j} \delta_{ij} \geq \frac{N(N-1)}{2} - \frac{N \cL_N}{4} + \frac 14
\left(\sum_i z_i\right) \left(\sum_i \partial_{z_i} \right) \,.
$$
The very last term is
non-negative, which yields (\ref{pap}) as a lower bound. Finally, one checks that for $2\leq L\leq N$ the lower bound is, in fact, an equality for the states $\mathcal{S} (z_1-z_{\rm CM})\cdots (z_L-z_{\rm CM})$, where $\mathcal{S}$ denotes symmetrization and $z_{\rm CM}:=N^{-1}\sum_{i=1}^N z_i$.

No exact formula for $\Delta_N(L)$ is known if $L > N$. For large
$N$ and $L\ll N^2$ the yrast line was studied in
\cite{LieSeiYng-09}, where it is proved that in this limit the
Gross-Pitaevskii energy becomes exact. The result in \cite{LieSeiYng-09} implies that in this regime the convex hull of $\Delta_N(L)$ is proportional to $N^3/L$. 

Very little is known about the yrast curve for $L\sim N^2$, in
particular concerning lower bounds. Upper bounds have been derived
using certain trial states (Pfaffian, composite fermions
\cite{CooWil-99,RegJol-04,RegChaJolJai-06}) which have been shown
numerically to have a large overlap with (some of) the true
eigenstates of the yrast curve, at least for small $N$. A rigorous
understanding of the properties of the true eigenstates is still
missing, however. It particular, it remains an open problem to
investigate whether $\liminf_{N\to\ii}d_1(N)>0$. This would imply
that the yrast curve has a discontinuous derivative at the Laughlin
state. It would also show a certain robustness of the Laughlin state,
in the sense that this state is the ground state for fixed
$\kappa>\kappa_1:=\limsup_{N\to\ii}\kappa_1(N)$, independently of the particle number $N$.

\medskip

This concludes our review of the properties of the effective
Hamiltonian (\ref{def:Hamil_LLL}). To state our last result, we will
denote by $P_N(\kappa)$ the (finite dimensional) orthogonal projector
in $L^2(\R^{3N})$ on the ground eigenspace of the operator $(4\pi
\kappa)^{-1}\cL_N/N^2+\Delta_N/N$, multiplied by $\pi^{-N/4}
e^{-\sum_{j=1}^N|\vx_j|^2/2}$. This orthogonal projector is constant
for all $\kappa\in(\kappa_{j+1}(N),\kappa_j(N))$. For
$\kappa>\kappa_1(N)$, it is just the projector on the $N$-body
Laughlin function
$$P_N(\kappa)=|\Psi^N_{\rm Lau}\rangle\langle\Psi^N_{\rm Lau}|$$
where
$$\Psi^N_{\rm Lau}(\vx_1,\dots,\vx_N)=k_N\pi^{-N/4}\prod_{1\leq i<j\leq N}(z_i-z_j)^2e^{-\sum_{k=1}^N|\vx_k|^2/2}.$$
The following theorem is a rather straightforward consequence of our
proof of Theorem~\ref{thm:bounds}.

\begin{thm}[{\bf Convergence of States and Fractional Quantum Hall
    Effect}]\label{cor:CV} Let $\kappa>0$ and $N\geq2$ fixed, and
  denote by $\Psi_{\omega,a}^N$ any chosen sequence of ground states
  of the Hamiltonian $H^N_{\omega,a}$ in $\bigvee_1^NL^2(\R^3)$. Then
  we have
\begin{equation}
\lim_{\substack{a\to0\\ a/(N\omega)\to\kappa}}\norm{\Psi^N_{\omega,a}-P_N(\kappa)\Psi^N_{\omega,a}}=0.
\end{equation}
In particular, if $\kappa>\kappa_1(N)$, one has $\Psi^N_{\omega,a}\to\Psi_{\rm Lau}^N$ when $a\to0$ and $a/(N\omega)\to\kappa$, up to a correct choice of a phase for $\Psi^N_{\omega,a}$.
\end{thm}

Theorem~\ref{cor:CV} shows that the ground state of any system of $N$
spinless trapped bosons with repulsive interactions of scattering
length $a$ is, for small enough $a$ and rotation speed close to the
critical one, well approximated by the ground state of the effective
Hamiltonian (\ref{def:Hamil_LLL}) on the LLL. In particular, the
states are highly correlated and exhibit a bosonic analogue of the
FQHE, with transitions between certain values of the angular momentum
coinciding with discontinuities of the derivative of the yrast curve.

The proof of Theorem~\ref{cor:CV} will be given Section \ref{proof:cor_CV}.
As will be obvious from the method of proof, both
Theorems~\ref{thm:bounds} and~\ref{cor:CV} can be extended to
low-lying excited energy eigenvalues and their corresponding
eigenfunctions as well. The corresponding analysis is similar to
previous studies of the effective one-dimensional behavior of Bose
gases in highly elongated traps \cite{SeiYin-08}, and we shall not
give the details here.

%%%%%%%%%%%%%%%%%%%%%%%%%%%%%%%%%%%%%%%%%%%%%
%%%%%%%%%%%%%%%%%%%%%%%%%%%%%%%%%%%%%%%%%%%%%

\medskip

\section{Proofs}
\subsection{Preliminaries}
In this subsection we shall gather some useful preliminary results
which will be needed in the rest of the proof.  We recall that
$\delta_{ij}$ was defined in \eqref{def:delta_ij}. Similarly the three-body
delta interaction can be defined on $\cB_3$ as
\begin{equation}
\left(\delta_{123}F\right)(z_1,z_2,z_3):=\frac{1}{(\sqrt{3}\pi)^3} F\left(\frac{z_1+z_2+z_3}{3},\frac{z_1+z_2+z_3}{3},\frac{z_1+z_2+z_3}{3}\right)
\end{equation}
where the prefactor was chosen such that
$$
\pscal{F,\delta_{123}F}_{\cB_3}= \left(\int_\R e^{-t^2} dt \right)^{-3} \int_\R dx \int_\C dz |F(z,z,z)|^2 e^{-3|z|^2-3x^2} \,.
$$
It defines a bounded self-adjoint operator on $\cB_3$. In fact, 
\begin{equation}\label{rel:delta}
\delta_{123}\leq \sqrt{\frac{2}{3\pi^3}}\, \delta_{12} \,.
\end{equation}

By definition all functions in $\gH_N$ are smooth. A way to quantify their regularity was provided by Carlen in \cite{Carlen-91}. We state it in the following lemma. 

\begin{lemma}[{\bf An inequality of Carlen \cite{Carlen-91}}]\label{lemma_Carlen} For any $p\in \N$, there exists a constant $C_{p}$ such that for any holomorphic function $f\in\cB_1$
\begin{equation}
\forall z\in\C,\quad \left|\frac{\partial^p f}{\partial z^p}(z)\right|^2e^{-|z|^2} \leq C_{p}(1+|z|^{2p}) \norm{f}_{\cB_1}^2.
\label{Carlen_inequality}
\end{equation}
\end{lemma}
\begin{proof}
The proof is a consequence of the Cauchy-Schwarz inequality and the following well-known \emph{coherent state representation} \cite{Bargmann-62,GirJac-84,AftBlaNie-06} for functions $f$ in the Bargmann space $\cB_1$:
\begin{equation}
f(z)=\pi^{-1}\int_{\C}f(\xi)\phi_\xi(z)e^{-|\xi|^2}d\xi
\label{coherent_state_representation}
\end{equation}
where $\phi_\xi(z)=e^{\overline{\xi}z}$.
\end{proof}

Using the smoothness of functions in the LLL, one can control any
interaction potential by a contact interaction, up to an error. This was done first in \cite{AftBlaLew}. 

\begin{lemma}[{\bf Controlling interaction potentials in the LLL}]\label{lem:comparison_delta}
Let $F\in\cB_N$ and $\Psi(\vx_1,\dots,\vx_N):=\pi^{-N/4} F(z_1,\dots,z_N)e^{-\sum_{j=1}^N|\vx_j|^2/2}\in\gH_N$. Let $g\in L^1(\R^3)$ be non-negative and radial. Then we have
\begin{multline}
\pscal{\Psi,\left(\sum_{1\leq i\neq j\leq N}g(\vx_i-\vx_j)\right)\Psi}_{\!\! L^2(\R^{3N})}\!\!\leq \left(\int_{\R^3}g\right)\pscal{F,\left(\sum_{1\leq i\neq j\leq N}\delta_{ij}\right)F}_{\!\!\cB_N}\\
+N^2C\left(\int_{\R^3}g(\vx)\frac{|\vx|^4}{1+|\vx|^4}\,dx\right)\norm{F}_{\cB_N}^2,
\label{estim_2_body}
\end{multline}
and
\begin{multline}
\pscal{\Psi,\left(\sum_{1\leq i\neq j\neq k\leq N}g(\vx_i-\vx_j)g(\vx_j-\vx_k)\right)\Psi}_{\!\! L^2(\R^{3N})}\\
\leq \left(\int_{\R^3}g\right)^2\pscal{F,\left(\sum_{1\leq i\neq j\neq k\leq N}\delta_{ijk}\right)F}_{\!\!\cB_N}\!\! + N^3C\left(\int_{\R^3}g(\vx)\frac{|\vx|^2}{1+|\vx|^2}dx \right)^2\norm{F}_{\cB_N}^2
\label{estim_3_body}
\end{multline}
for a universal constant $C>0$.
\end{lemma}

\begin{proof}
An inequality similar to \eqref{estim_2_body} was derived before in \cite{AftBlaLew}. We shall only write the proof of \eqref{estim_2_body} for $N=2$ (the general case is then obtained by summing over pairs). We shall omit the proof of \eqref{estim_3_body} which is analogous.

Defining 
$G(u,v)=F\left(\frac{u+v}{\sqrt{2}},\frac{u-v}{\sqrt{2}}\right)$, we have  
\begin{align}\nonumber
\pscal{\Psi,g(\vx_1-\vx_2)\Psi}&=\frac 1\pi \int_{\R^3}\int_{\R^3}g(\vx_1-\vx_2)\,|F(z_1,z_2)|^2e^{-|\vx_1|^2-|\vx_2|^2}dx_1\,dx_2\\
&=\int_{\C}\int_{\C}\tilde{g}(\sqrt{2}|v|)\,|G(u,v)|^2e^{-|u|^2-|v|^2}du\,dv \,, \label{eq:introg}
\end{align}
where we have introduced
$$
\tilde{g}(|z|):= \frac 1\pi \int_{\R}\int_\R g(z,x^3_1-x^3_2)e^{-|x^3_1|^2-|x^3_2|^2}dx^3_1\,dx^3_2$$
which is obviously radial, i.e., depends only on $|z|$. 

We split the $v$ integral in (\ref{eq:introg}) into two parts, corresponding to $|v|\leq 1$ and $|v|\geq 1$, respectively. Consider first the case $|v|\leq 1$. 
Using the fact that $G(u,\cdot)$ is even  because $F$ is symmetric, a Taylor expansion yields
$$G(u,v)=G(u,0)+v^2\int_0^1(1-t)\frac{\partial^2 G}{\partial v^2}(u,tv)\,dt \,.$$
By the radiality of $\tilde{g}$  the cross term vanishes when integrating over angles, and hence 
\begin{multline}
 \int_{|v|\leq 1} \tilde{g}(\sqrt{2}|v|)\,|G(u,v)|^2e^{-|v|^2} dv=   \int_{|z|\leq 1} \tilde{g}(\sqrt{2}|v|)\,|G(u,0)|^2e^{-|v|^2}dv \\
+\int_{|v|\leq 1}\tilde{g}(\sqrt{2}|v|)\,|v|^4\left|\int_0^1(1-t)\frac{\partial^2 G}{\partial v^2}(u,tv)\,dt\right|^2e^{-|v|^2}dv\,.
\label{formula_interaction_LLL}
\end{multline}
We integrate this identity against  $e^{-|u|^2}du$. The first term becomes 
$$
\left(\int_{|z|\leq 1/\sqrt{2}}g(z,x^3)e^{-|\vx|^2/2}dx\right)\pscal{F,\delta_{12}F}_{\cB_2} \leq \left(\int_{\R^3} g\right) \pscal{F,\delta_{12}F}_{\cB_2}\,.
$$
With the aid of Carlen's inequality (\ref{Carlen_inequality}), the second term is bounded above by 
\begin{equation*}
C\int_{|z|\leq 1/\sqrt{2}}g(z,x^3)\,|z|^4 \,dx\;\norm{F}_{\cB_2}^2 \leq C'\int_{\R^3} g(\vx) \frac{|\vx|^4}{1+|\vx|^4}\,dx\;\norm{F}_{\cB_2}^2 \,.
\end{equation*}

Finally, for $|v|\geq 1$ we shall use again (\ref{Carlen_inequality}), this time for $p=0$, to conclude that 
\begin{align*}
\int_{\C}du \int_{|v|\geq 1} dv\, \tilde{g}(\sqrt{2}|v|)\,|G(u,v)|^2e^{-|u|^2-|v|^2}
& \leq C \int_{|z|\geq 1/\sqrt{2}} g(z,x^3) dx \;\norm{F}_{\cB_2}^2 
\\ & \leq C'\int_{\R^3} g(\vx) \frac{|\vx|^4}{1+|\vx|^4}\,dx\;\norm{F}_{\cB_2}^2\,.
\end{align*}
This completes the proof.
\end{proof}

\begin{remark}\label{rmk:lower_bound_lemma}\it 
Although we will not need it, we note that \eqref{formula_interaction_LLL} also yields a lower bound:
$$\pscal{\Psi,g(\vec{x}_1-\vec{x}_2)\Psi}_{\gH_2}\geq \left(\int_{\R^3}g(\vec{x})e^{-\frac{|\vec{x}|^2}{2}}\, d{x}\right)\pscal{F,\delta_{12}F}_{\cB_2},$$
where we have used the same notation as in Lemma \ref{lem:comparison_delta}.
Combined with \eqref{estim_2_body} this shows that the restriction of the operator $\epsilon^{-3}g\big((\vec{x}_1-\vec{x}_2)/\epsilon\big)$ to the LLL converges to $(\int_{\R^3}g)\delta_{12}$ as $\epsilon\to0$.
\end{remark}

It will be important to have some \emph{a priori} bounds on the
ground state energy of the effective Hamiltonian (\ref{def:Hamil_LLL}). The
following is certainly not optimal but it has the merit of being
simple.

\begin{lemma}[{\bf Simple bounds on $E^{\rm LLL}_N(\omega,a)$}]\label{lem:bounds_on_E_LLL}
We have, for $\kappa=a/(N\omega)$ and $N\geq 2$,
\begin{equation}
aN\,C\min\left\{\frac{1}{\kappa N},N\right\}\leq E^{\rm LLL}_N(\omega,a)\leq aN\,\min\left\{\frac{1}{\kappa},\sqrt{\frac 2\pi}N\right\}\,.
\label{estim_energy_LLL} 
\end{equation}
\end{lemma}
\begin{proof}
The upper bound is obtained by taking as trial states the Laughlin
  function and the constant function, respectively. For the lower
  bound, we note that
\begin{multline*}
 \sum_{1\leq i<j\leq N}\left(\frac{\omega}{N-1}(L_{z_i}+L_{z_j})+4\pi a\delta_{ij}\right)\\
\geq c\min\left\{\omega N/2\,,\, 2\pi aN(N-1)\right\}\geq\frac{c}2\min\left\{\kappa^{-1}a\,,\, aN^2\right\}
\end{multline*}
where $c=\inf\sigma_{\cB_2}(L_1+L_2+\delta_{12})>0$.
\end{proof}

Except for the prefactor, the upper bound is expected to be sharp. In other words, the lower bound should hold without the factor $1/N$ multiplying $\kappa^{-1}$, for an appropriate constant $C$. This remains an open problem, however.

\medskip

%%%%%%%%%%%%%%%%%%%%%%%%%%%%%%%%%%%%%%
\subsection{Proof of Theorem \ref{thm:bounds}}\label{sec:proof_Thm}

\subsubsection*{\textbf{Step 1: Upper Bound}}
We start by proving the upper bound, using the variational
principle. The main difficulty is to get the scattering length in
front of the interaction potential. As suggested first by Dyson in
\cite{Dyson-57}, this is done by multiplying a trial state $\Phi$ of
$\gH_N$ by a correlated function $S$ accounting for the short scale
structure of the ground state. Compared to
previous similar arguments in
\cite{LieSeiYng-00,Sei-03}, a new complication comes from the that the trial state $\Phi$ in $\gH_N$
is not a simple product function, but itself already a (possibly)
highly correlated state of which little is known. Fortunately,
the information that $\Phi\in\gH_N$ combined with simple bounds on
$E^{\rm LLL}_N(\omega,a)$ will allow to get the desired upper bound.

Let $F_{N,\omega,a}\in\cB_N$ be a normalized ground state for the LLL
Hamiltonian $\tilde{H}^N_{\omega,a}$ defined in \eqref{def:Hamil_LLL},
which is a common eigenvector of $\cL_N$ and $\Delta_N$. We consider
the following trial state:
\begin{equation}
\Psi_{N,\omega,a}:=S_{N,a}\, \Phi_{N,\omega,a}
\label{trial_state}
\end{equation}
where 
$$S_{N,a}(\vx_1,\dots,\vx_N):=\prod_{1\leq i< j\leq N} f_a(|\vx_i-\vx_j|)$$
for some $0\leq f_a\leq1$ which will be defined later, and
$$\Phi_{N,\omega,a}(\vx_1,\dots,\vx_N):=\pi^{-N/4} F_{N,\omega,a}(z_1,\dots,z_N)e^{-\sum_{i=1}^N\frac{|\vx_i|^2}{2}}.$$
Note that the norm of $\Phi_{N,\omega,a}\in L^2(\R^{3N})$ equals the norm of $F_{N,\omega,a}\in \cB_N$.
We write 
$$H^N_{\omega,a}=\sum_{j=1}^N(h_{\omega})_{j}+\sum_{i<j}W_a(\vx_i-\vx_j)$$
where
\begin{equation}
h_{\omega}=\frac{|\vp-\ve_3\times \vx|^2}2+\frac{|x^3|^2}{2}-\frac32+\omega \ve_3\cdot \vL \,.
\label{one-body_2a}
\end{equation}
We shall also use the notation $H_{\omega,0}^N = \sum_{j} (h_\omega)_j$ for short. 
Using the fact that $f_a$ is real, we can argue as in \cite[Eq.~(4.64)]{Sei-03} to get the identity
$$
\pscal{\Psi_{N,\omega,a},H_{\omega,0}^N\Psi_{N,\omega,a}}=\sum_{j=1}^N\int \frac{|\vec\nabla_{j}S_{N,a}|^2}{2}|\Phi_{N,\omega,a}|^2
+\Re\pscal{S_{N,a}^2\Phi_{N,\omega,a}\,,\,H_{\omega,0}^N\Phi_{N,\omega,a} }.
$$
Using \eqref{eq:angular_momentum} and that $F_{N,\omega,a}$ is a normalized eigenvector of $\cL_N$,  we have
$$
H_{\omega,0}^N \Phi_{N,\omega,a}=\omega\pscal{F_{N,\omega,a},\cL_N F_{N,\omega,a}}_{\cB_N}\Phi_{N,\omega,a} \,,$$
hence
\begin{multline*}
\pscal{\Psi_{N,\omega,a},H_{\omega,0}^N \Psi_{N,\omega,a}}=\sum_{j=1}^N\int \frac{|\vec\nabla_{j}S_{N,a}|^2}{2}|\Phi_{N,\omega,a}|^2\\
+\norm{\Psi_{N,\omega,a}}^2\pscal{\Phi_{N,\omega,a},H_{\omega,0}^N\Phi_{N,\omega,a}}\,.
\end{multline*}
We compute
\begin{equation*}
\vec\nabla_{k}S_{N,a}=\sum_{i\neq k}f_a'(|\vx_k-\vx_i|)\frac{\vx_k-\vx_i}{|\vx_k-\vx_i|}\prod_{\substack{1\leq m<n\leq N\\ \{m,n\}\neq\{i,k\}}}f(|\vx_m-\vx_n|)\,.
\end{equation*}
Using $0\leq f\leq 1$ we therefore get
\begin{equation*}
 \frac12\sum_{k=1}^N|\vec\nabla_{k}S_{N,a}|^2 \leq \sum_{1\leq i<j\leq N}f_a'(|\vx_i-\vx_j|)^2
+ \frac12\sum_{i\neq j\neq k}f_a'(|\vx_i-\vx_j|)f_a'(|\vx_k-\vx_j|).
\end{equation*}
We finally deduce that
\begin{multline*}
\pscal{\Psi_{N,\omega,a},H^N_{\omega,a}\Psi_{N,\omega,a}}\leq \norm{\Psi_{N,\omega,a}}^2\pscal{\Phi_{N,\omega,a},H^N_{\omega,0}\Phi_{N,\omega,a}}\\
+\pscal{\Phi_{N,\omega,N,a},\left(\sum_{1\leq i<j\leq N}\big[(f_a')^2+W_af_a^2\big](|\vx_i-\vx_j|)\right)\Phi_{N,\omega,N,a}}\\
+\frac12\pscal{\Phi_{N,\omega,N,a},\left(\sum_{i\neq j\neq k}f_a'(|\vx_i-\vx_j|)f_a'(|\vx_k-\vx_j|)\right)\Phi_{N,\omega,N,a}}.
\end{multline*}

The next step is to bound the terms on the right hand side of the
previous inequality. An essential tool is the inequality
\eqref{estim_2_body} of Lemma \ref{lem:comparison_delta} which
relates, on the LLL, the interaction of a smooth potential
$\sum_{i<j}g(\vx_i-\vx_j)$ with that of the contact interaction with
coefficient $\int_{\R^3} g$. As we will see, for a correct choice of
$f_a$, we will have $\int[(f_a')^2+W_af_a^2]\simeq 4\pi a$, as
desired. Let $g_a:=(f'_a)^2+W_a(f_a)^2$. Using
Lemma~\ref{lem:comparison_delta} as well as the bound
(\ref{rel:delta}), we obtain
\begin{multline*}
\pscal{\Psi_{N,\omega,a}\,,H^N_{\omega,a}\,\Psi_{N,\omega,a}} \leq E^{\rm LLL}_N(\omega,a)\norm{\Psi_{N,\omega,a}}^2\\
+\left(\int_{\R^3}g_a-4\pi a\norm{\Psi_{N,\omega,a}}^2+\sqrt{\frac{2}{3\pi^3}} N\left(\int f_a'\right)^2\right)\pscal{ F_{N,\omega,a}\,,\Delta_N\,F_{N,\omega,a}}_{\cB_N}\\
+CN^2\int_{\R^3}g_a(\vx)\frac{|\vx|^4}{1+|\vx|^4}dx+CN^3\left(\int_{\R^3}f_a'(|\vx|)\frac{|\vx|^2}{1+|\vx|^2}dx\right)^2.
\end{multline*}
For an upper bound, we can use $\pscal{ F_{N,\omega,a}\,,\Delta_N\,F_{N,\omega,a}}_{\cB_N}\leq E_N^{\rm LLL}(\omega,a)/(4\pi a)$, hence
\begin{multline*}
\pscal{\Psi_{N,\omega,a}\,,H^N_{\omega,a}\,\Psi_{N,\omega,a}} \leq \frac{E^{\rm LLL}_N(\omega,a)}{4\pi a} \left(\int_{\R^3}g_a+\sqrt{\frac{2}{3\pi^3}} N\left(\int f_a'\right)^2\right)\\
+CN^2\int_{\R^3}g_a(\vx)\frac{|\vx|^4}{1+|\vx|^4}dx+CN^3\left(\int_{\R^3}f_a'(|\vx|)\frac{|\vx|^2}{1+|\vx|^2}dx\right)^2.
\end{multline*}

Let us now choose $f_a$. As in \cite{LieSeiYng-00,LieSeiSolYng-05} we
take, for some $b>a$ to be specified later,
\begin{equation}\label{def:fa}
f_a(s):=\left\{\begin{array}{ll}\frac{u_a(s)/s}{u_a(b)/b}&\text{if $0\leq s\leq b$}\\ 1 & \text{if $s\geq b$}\end{array}\right.
\end{equation}
where $u_a$ is the solution of the scattering equation
$$-u_a''(s)+W_a(s)u_a(s)=0$$
with $u_a(0)=0$ and $\lim_{s\to\ii}u'_a(s)=1$. Integrating by parts and using that $0\leq u_a(s)\leq s$ and $0\leq s u_a'(s) - u_a(s) \leq a$ \cite{LieSeiYng-00,LieSeiSolYng-05} we see that  
$$
\int_{\R^3}g_a=\int_{\R^3}[(f_a')^2+W_af_a^2]\leq \frac{4\pi a}{1-a/b}\,.
$$
By splitting the integral into a part $|x|\leq (a b^3)^{1/4}$ and $|x|\geq (a b^3)^{1/4}$, one checks that
$$
\int_{\R^3}g_a(|\vx|)\frac{|\vx|^4}{1+|\vx|^4}\,dx\leq \frac{C}{1-a/b}\left( a^2 b^3  + a b^4 \int_{|\vx|\geq (b/a)^{3/4}} W(\vx)  dx \right)\,. 
$$
Note that for $a\ll b$ the second term in the last bracket is small compared to the first one if $W$ decays faster than $|\vx|^{-3-4/3}$ at infinity.
We further have 
$$
\int_{\R^3}f'_a\leq \frac{4\pi ab}{1-a/b}\,,
$$
and hence 
$$
\int_{\R^3}f'_a(|\vx|)\frac{|\vx|^2}{1+|\vx|^2}\,dx\leq \frac{4\pi ab^3}{1-a/b}\,.
$$

Let us assume, for simplicity, that $b \geq 2a $. Then
\begin{multline*}
\pscal{\Psi_{N,\omega,a}\,,H^N_{\omega,a}\,\Psi_{N,\omega,a}} \leq E^{\rm LLL}_N(\omega,a) \left( 1+ C\left[ \frac a b + N a b^2 \right] \right)\\
+C N^2 a^2 b^3 \left(1+ N b^3+ \frac ba \int_{|\vx|\geq (b/a)^{3/4}} W(\vx) dx \right)
\end{multline*}
for some constant $C>0$. To bound the last term relative to the first one, we can use the lower bound of Lemma~\ref{lem:bounds_on_E_LLL} to conclude that
\begin{multline}\label{eq:last_estimate_upper_bound}
\pscal{\Psi_{N,\omega,a}\,,H^N_{\omega,a}\,\Psi_{N,\omega,a}} \leq E^{\rm LLL}_N(\omega,a) \biggl( 1+ C \biggl[ \frac a b + N a b^2  \\ +  \frac{\kappa N^2 a b^3}{\min\{1,\kappa N^2\}} \biggl(1+ N b^3+ \frac ba \int_{|\vx|\geq (b/a)^{3/4}} W(\vx)  dx \biggl) \biggl]\biggl)\,.
\end{multline}

It remains to derive a lower bound on
$\norm{\Psi_{N,\omega,a}}$. Arguing as in \cite{LieSeiYng-04} we can
bound
\begin{align*}\nonumber
\norm{\Psi_{N,\omega,a}} ^2&=\int_{\R^3}\cdots\int_{\R^3}\prod_{1\leq i<j\leq N}f_a(|\vx_i-\vx_j|)^2|\Phi_{N,\omega,a}|^2\\ \nonumber 
&\geq 1-\sum_{1\leq i<j\leq N}\int_{\R^3}\cdots\int_{\R^3}(1-f_a^2)(|\vx_i-\vx_j|)|\Phi_{N,\omega,a}|^2\\ \nonumber
&\geq 1-CN^2\int_{\R^3}(1-f_a^2)(|\vx|)\frac{|\vx|^4}{1+|\vx|^4}\,dx\\
&\qquad\qquad\qquad\qquad-\left(\int_{\R^3}(1-f_a^2)\right)\pscal{F_{N,\omega,a},\Delta_N F_{N,\omega,a}}_{\cB_N}
\end{align*}
where we have used again Lemma \ref{lem:comparison_delta}. For our
choice of $f_a$ in (\ref{def:fa}), it is easy to see
\cite{LieSeiYng-00,LieSeiSolYng-05} that
$$
\int_{\R^3}(1-f_a^2)\leq 4\pi a b^2
$$
and hence 
$$
\int_{\R^3}(1-f_a^2)(|\vx|)\frac{|\vx|^4}{1+|\vx|^4}\,dx\leq 4\pi a b^6  \,.
$$
Finally, we can use Lemma~\ref{lem:bounds_on_E_LLL} to bound
$\pscal{F_{N,\omega,a},\Delta_N F_{N,\omega,a}}_{\cB_N}\leq E_N^{\rm
  LLL}(\omega,a)/(4\pi a) \leq (2\pi)^{-3/2} N \kappa^{-1}$.
This yields
\begin{equation}
\norm{\Psi_{N,\omega,a}} ^2 \geq  1-CN a b^2 \left( \frac 1\kappa +  N b^4 \right) \label{eq:last_ub2}
\end{equation}
for an appropriate constant $C>0$. 

Combining (\ref{eq:last_estimate_upper_bound}) and (\ref{eq:last_ub2}), the choice $b=2\kappa^{-1/4}N^{-1/2}$ leads to the desired inequality \eqref{eq:upper_bound}.

\begin{remark}\it
If we had a lower bound 
$$E^{\rm LLL}_N(\omega,a)\geq c\, Na \kappa^{-1}\, ,$$
as is expected for $\kappa\gtrsim N^{-1}$, the upper bound of
Theorem~\ref{thm:bounds} could be somewhat improved. In
(\ref{eq:last_estimate_upper_bound}) $\min\{1,\kappa N^2\}$ could be
replaced by $N \min\{1,\kappa N\}$ in the denominator, and the optimal
choice of $b$ would then be $b=(\kappa N)^{-1/3}$. The main error term
would then be of the order $a\kappa^{1/3}N^{1/3}$ instead of
$a\kappa^{1/4}N^{1/2}$.
\end{remark}

\medskip

%%%%%%%%%%%%%%%%%%%%%%%%%%%%%%%%%%%%%%%%%%%%%
\subsubsection*{\textbf{Step 2: Lower Bound}}
As a first step, we shall replace $W_a$ by the finite range potential
$W_{a,R_0}=W_a\chi(|\vx|\leq R_0)$ for some $R_0$ to be chosen
later. Since $W_a$ is assumed to be non-negative, this is legitimate
for a lower bound.  We denote by $a(R_0)$ the scattering length of
$W_{a,R_0}$. If $R_0$ is large enough compared to $a$, we will have
$a(R_0)\simeq a$. Indeed, let us recall that \cite{LieSeiSolYng-05}
\begin{equation}
  4\pi a\geq 4\pi a(R_0)=\int_{\R^3}W_{a,R_0}f_{a,R_0}\geq \int_{\R^3}W_{a,R_0}f_a\geq 4\pi a-\int_{|\vx|\geq R_0}W_a
\label{estim_scattering_lower_bound} 
\end{equation}
where $f_a \leq f_{a,R_0}$ are the solutions of the zero-scattering
equations corresponding to $W_a$ and $W_{a,R_0}$, respectively.

We continue with a lemma inspired by a method of Dyson
\cite{Dyson-57}. The key idea is to replace the \lq\lq hard\rq\rq\
interaction potential $W_a$ by a softer one using parts of the kinetic
energy, with this softer potential being close to $4\pi a \delta$ when
projected to the LLL. Note that this step is essential, it is not
possible to project the original $W_a$ to the LLL level. This would
also lead to a $\delta$ interaction for small $a$, but with the wrong
coupling constant $\int W_a$ instead of $4\pi a$ (see Remark \ref{rmk:lower_bound_lemma}). Compared to previous
studies \cite{LieYng-98,LieSeiSol-05,LieSei-06,SeiYin-08b} where a
similar strategy has been applied, the main new difficulty comes from
the fact that the effective kinetic energy $h_0= (|\vp-\ve_3\times
\vx|^2 + |x^3|^2 - 3)/2$ is not positive {\it locally}, i.e., on a
domain with Neumann boundary conditions, but is positive only on the
whole of $\R^3$. To circumvent this problem, we shall rewrite
$\pscal{\Psi,h_0\Psi}$ for $\Psi\in L^2(\R^3)$ as
\begin{equation}\label{rewr}
\pscal{\Psi,h_0\Psi} = \frac 12 \int_{\R^3} e^{-|\vx|^2} \left( |\partial_{x^3} \psi(\vx)|^2 +
    |\partial_{\bar z} \psi(\vx)|^2 \right) dx
\end{equation}
with $\psi(\vx) = e^{|\vx|^2/2} \Psi(\vx)$ and $\partial_{\bar z}
= \partial_{x^1} + i \partial_{x^2}$. The integrand on the right side
is now positive but contains no derivatives  with respect to $z$ and is
hence weaker than $|\vec\nabla \psi|^2$. Nevertheless we shall show in the
next lemma that it is still strong enough to accomplish the goal of
replacing $W_a$ by a softer potential for a lower bound. The resulting
\lq\lq potential\rq\rq\ turns out not be a potential in the usual
sense of a multiplication operator, but rather is a non-local operator
which has the property that its projection to the LLL is proportional
to a $\delta$-function, however.

\begin{lemma}[{\bf Dyson-type inequality}]\label{lem:Dyson}
Let  $\vy=(s,y^3)\in \R^3$. For $R>R_0$, we have for all $\psi$
\begin{multline}\label{dysl}
  \int_{|\vx-\vy|\leq R} e^{-|\vx|^2} \left( |\partial_{x^3} \psi(\vx)|^2 +
    |\partial_{\bar z} \psi(\vx)|^2 + W_a(\vx-\vy) |\psi(\vx)|^2 \right)dx \\
  \geq 4\pi a(R_0)\, e^{-(|y^3|+R)^2 + |s|^2} \left| \frac 1{4\pi R^2}
    \int_{|\vx-\vy|= R} e^{-\bar s z} \psi(\vx)\, dx \right|^2\,.
\end{multline}
\end{lemma}

Note that if $\psi\in\cB_1$, one has
$$\frac 1{4\pi R^2} \int_{|\vx-\vy|= R} e^{-\bar s z} \psi(\vx)\, dx=e^{-|s|^2}\psi(\vy)\,,$$
and the right side of (\ref{dysl}) equals $4\pi a(R_0) e^{-(|y^3|+R)^2
  - |s|^2} |\psi(\vy)|^2$ which is precisely $4\pi
a(R_0)\pscal{\psi,\delta_y\psi}_{\cB_1}$ when $R=0$.

\begin{proof}[Proof of Lemma \ref{lem:Dyson}]
Let $g(\vx) = e^{-\bar s z}\psi(\vx +\vy)$. Using that $W_a\geq W_{a,R_0}$, we have to show that
\begin{multline}
  \int_{|\vx|\leq R} e^{-|z|^2-(x^3+y^3)^2} \left( |\partial_{x^3} g(\vx)|^2 + |\partial_{\bar z} g(\vx)|^2 + W_{a,R_0}(\vx) |g(\vx)|^2 \right)\,dx \\
  \geq 4\pi a(R_0)\, e^{-(|y^3|+R)^2 } \left| \frac 1{4\pi R^2}
    \int_{|\vx|=R} g(\vx)\, dx \right|^2\,.
\end{multline}
Since  $|z|^2 +(x^3+y^3)^2 \leq (|y^3| + R)^2$ in the integrand on the left, this will follow if we can show that
\begin{multline}\label{toshow}
  \int_{|\vx|\leq R}  \left( |\partial_{x^3} g(\vx)|^2 + |\partial_{\bar z} g(\vx)|^2 + W_{a,R_0}(\vx) |g(\vx)|^2 \right)dx\\
  \geq 4\pi a(R_0) \left| \frac 1{4\pi R^2} \int_{|\vx|=R} g(\vx)\, dx
  \right|^2\,.
\end{multline}
Let now $f_{a,R_0}$ be the solution of the zero-energy scattering
equation $-\Delta f_{a,R_0} + W_{a,R_0} f_{a,R_0} = 0$, subject to the
normalization $\lim_{|\vx|\to \infty} f_{a,R_0}(\vx) = 1$. Since
$f_{a,R_0}$ is real-valued, $|\partial_{\bar z} f_{a,R_0}|^2 =
|\partial_{x^1} f_{a,R_0}|^2 + |\partial_{x^2} f_{a,R_0}|^2$, and
hence
$$
\int_{|\vx|\leq R}  \left( |\partial_{x^3} f_{a,R_0}(\vx)|^2 + |\partial_{\bar z} f_{a,R_0}(\vx)|^2 + W_{a,R_0}(\vx) |f_{a,R_0}(\vx)|^2 \right)dx  \leq 4\pi a(R_0) \,.
$$
The Cauchy-Schwarz inequality implies that 
\begin{multline}
\int_{|\vx|\leq R}  \left( |\partial_{x^3} g(\vx)|^2 + |\partial_{\bar z} g(\vx)|^2 + W_{a,R_0}(\vx) |g(\vx)|^2 \right)\,dx\\ 
\geq  \frac 1{4\pi a(R_0)}\left|\int_{|\vx|\leq R}  \big( \partial_{x^3} f_{a,R_0} \partial_{x^3} g+ \partial_{z}f_{a,R_0} \partial_{\bar z} g + W_{a,R_0} f_{a,R_0} g \big)  \right|^2\,. 
\end{multline}
Using partial integration, the zero-energy scattering equation as well
as the fact that $|\nabla f_{a,R_0}(\vx)|= a/|\vx|$ for $|\vx|\geq
R_0$ this yields (\ref{toshow}).
\end{proof}

As an immediate corollary, we see that for any non-negative function
$\rho$ supported on $[R_0,R]$ with $\int_{R_0}^R \rho\leq 1$,
\begin{multline}\label{dyslc}
  \int_{|\vx-\vy|\leq R} e^{-|\vx|^2} \left( |\partial_{x^3} \psi(\vx)|^2 +
    |\partial_{\bar z} \psi(\vx)|^2 + W_a(\vx-\vy) |\psi(\vx)|^2 \right)dx \\
  \geq 4\pi a(R_0)\, e^{-(|y^3|+R)^2 + |s|^2} \int_{R_0}^R dr\,
  \rho(r) \left| \frac 1{4\pi r^2} \int_{|\vx-\vy|=r} e^{-\bar s z} \psi
  \right|^2\,.
\end{multline}
We shall apply this inequality to the Hamiltonian $H^N_{\omega,a}$,
for each particle separately, considering the other $N-1$ particles as
fixed. Consider first particle one, and assume that all particles
$k\geq 2$ are located at a distance $\geq2R$ from each other, i.e.,
that $|\vx_k-\vx_\ell|\geq 2R$ for all $k,\ell=2,\dots,N$. Then we get,
for all functions $F(\vx_1,\dots,\vx_N)$,
\begin{align*}\label{dysld}
&\int_{\R^3}dx_1\; e^{-|\vx_1|^2} \Bigg( |\partial_{x_1^3} F(\vx_1,\dots,\vx_N)|^2
+ |\partial_{\bar z_1} F(\vx_1,\dots,\vx_N)|^2 \\
&\qquad\qquad\qquad\qquad\qquad\qquad+ \sum_{j=2}^NW_a(\vx_j- \vx_1) |F(\vx_1,\dots,\vx_N|^2 \Bigg) \\
& \geq 4\pi a(R_0)\sum_{j=2}^N e^{-(|x_j^3|+R)^2 + |z_j|^2}\!\! \int_{R_0}^R dr\,\! \rho(r) \!\left| \frac 1{4\pi r^2} \int_{|\vx_1-\vx_j|=r}\!\!\! e^{-\overline{z_j} z_1}F(\vx_1,\dots,\vx_N)\, dx_1 \right|^2\,.
\end{align*}
In general, we can get the same bound if we only retain on the right
side the $\vx_j$'s for $j=2,\dots,N$ which are at a distance $\geq 2R$ from
all the others. Using (\ref{one-body_2a}) and (\ref{rewr}), we
conclude that, for any $0\leq \theta\leq 1$,
\begin{equation}\label{lbh}
H^N_{\omega,a} \geq  \sum_{j=1}^N \big(  \theta \left(h_0\right)_j+ \omega(\ve_3\cdot \vL)_j \big) + 4\pi a(R_0) (1-\theta) \sum_{1\leq i< j\leq N} U_{ij} \,,
\end{equation}
where the potential $U_{ij}$ is defined as
\begin{multline}
\pscal{\Psi,U_{12}\Psi} = \int dx_2\cdots\int dx_N\, e^{-|\vx_2|^2} e^{-(|x_2^3|+R)^2 + |z_2|^2}   e^{-\sum_{k=3}^N|\vx_k|^2}  \times\\ 
\times \prod_{k=3}^N \chi_{|\vx_k-\vx_2|\geq 2R}\int_{R_0}^R dr\, \rho(r) \left| \frac 1{4\pi r^2} \int_{|\vx_1-\vx_2|=r} e^{-\overline{z_2} z_1} F(\vx_1,\dots,\vx_N)\, dx_1  \right|^2,
\end{multline}
with $F(\vx_1,\dots,\vx_N) = \pi^{N/4} \Psi(\vx_1,\dots,\vx_2) \prod_{j=1}^N e^{|\vx_j|^2/2}$ and $\vx_1=(z_1,x_1^3)$, $\vx_2=(z_2,x_2^3)$.

The new potential $\sum_{i<j} U_{ij}$ is a complicated $N$-body term which has the advantage of being bounded, however. To be precise, the following bound holds.

\begin{lemma}\label{lem:bu}
We have
\begin{equation}\label{normu}
\| U_{12} \| \leq \sup_{r} \frac{\rho(r)}{4\pi r^2}\,.
\end{equation}
\end{lemma}

\begin{proof}
Applying the Cauchy-Schwarz inequality to the $x_1$ integration, we see that
$$
\left| \int_{|\vx_1-\vx_2|=r}d x_1\, e^{-\overline{z_2} z_1} \psi \right|^2 \leq  \int_{|\vx_1-\vx_2|=r}d x_1\, e^{-|\vx_1|^2} |\psi|^2 \   \int_{|\vx_1-\vx_2|=r}dx_1\,  e^{| \vx_1|^2 - 2 {\rm Re\,} \bar z_2 z_1}\,.
$$
It is easy to check that $-(|x_2^3|+R)^2 + |z_2|^2 + | \vx_1|^2 - 2 {\rm Re\,} \bar z_2 z_1 \leq 0$ for $| \vx_1- \vx_2| \leq R$. Hence 
$$
\pscal{\Psi, U_{12} \Psi} \leq \int\cdots \int |\Psi(\vx_1,\dots,\vx_N)|^2 \frac {\rho(|\vx_1-\vx_2|)}{4\pi | \vx_1-\vx_2|^2}\,dx_1\cdots dx_N\,.
$$
This proves the claim.
\end{proof}

In order to minimize the right side of (\ref{normu}), we shall choose
\begin{equation}\label{def:rho}
\rho(r) = \frac { 3 r^2}{R^3 - R_0^3} \quad {\rm for} \quad R_0\leq r\leq R.
\end{equation}
Note that $\int_{R_0}^R\rho(r)dr=1$.

\medskip

We shall now apply a standard perturbation theory argument to our
Hamiltonian. Let $P$ denote the orthogonal projection onto $\gH_N$
(the LLL for all the $N$ particles), and let $Q=1-P$. Let
$$
A = \sum_{j=1}^N \left(  \theta \left(h_0\right)_j + \omega(\ve_3\cdot \vL)_j \right)
$$
and 
$$
B =  4\pi a(R_0) (1-\theta) \sum_{i< j} U_{ij}
$$
so that the right side of (\ref{lbh}) equals $A+B$. We have $A = PAP + QAQ$. Since $B$ is positive we can use the Cauchy-Schwarz inequality to get the lower bound
$$
B \geq (1-\delta) PBP + \left( 1 - \delta^{-1} \right) QBQ
$$
for any $0<\delta<1$. Using Lemma~\ref{lem:bu} with the choice (\ref{def:rho}) for $\rho$  we see that
$$
QBQ \leq Q \, \frac {6\pi a N(N-1)}{R^3 - R_0^3}\,.
$$ 
(Recall that $a(R_0)\leq a$.) 
Moreover, for $\theta> \omega$, 
$$
QAQ \geq Q \, \left( \theta  - \omega\right)\,.  
$$
We therefore conclude that 
\begin{equation}
A+ B \geq (1-\delta) P (A+B) P + Q \left( \theta -  \omega -\frac {6\pi a N^2}{\delta(R^3 - R_0^3)}\right)\,. 
\label{estim:A_plus_B} 
\end{equation}
In particular, 
$$
E_N(\omega,a ) \geq  \min\left\{ (1-\delta)\, \inf\sigma (A+B)\restriction_{\gH_N}\, , \,  \theta - \omega -\frac {6\pi a N^2}{\delta(R^3 - R_0^3)}\right\}
$$
and it remains to study $A+B$ restricted to $\gH_N$.

For $\Psi(\vx_1,\dots,\vx_N) = \pi^{-N/4} \prod_{j=1}^N
e^{-|\vx_j|^2/2}F(z_1,\dots,z_N)$ a bosonic function in $\gH_N$, we have
$$
\pscal{\Psi,A\Psi} = 
\omega\pscal{F,\cL_N F}_{\cB_N}
$$
by \eqref{eq:angular_momentum}, and 
\begin{multline}
\pscal{\Psi,B\Psi} = 2\pi a(R_0) N(N-1) (1-\theta)  \int d x_2\, e^{|\vx_2|^2 - (|x_2^3|+R)^2}\\
\times  \int\prod_{j\geq 3}\chi_{|\vx_j- \vx_2|\geq 2R}\, dx_j \, |\Psi( \vx_2, \vx_2,  \vx_3,\dots, \vx_N)|^2
\end{multline}
where we have used that $F$ is analytic in $z_1$ and $z_2$.
For a lower bound, we use 
$$
\prod_{j\geq 3} \chi_{|\vx_j- \vx_2|\geq 2R} \geq 1 - \sum_{j\geq 3}  \chi_{| \vx_j- \vx_2|\leq 2R}\,.
$$
Letting  
$$
t_R = \frac 2{\sqrt\pi}  \int_0^\infty dt \, e^{- (t+R)^2} \geq 1 - \frac {2R}{\sqrt\pi}\,,
$$
we finally get 
\begin{multline}
\pscal{\Psi,B\Psi}\geq 4\pi a(R_0)  (1-\theta) t_R \pscal{F,\left(\sum_{1\leq i<j\leq N}\delta_{ij}\right)F}_{\cB_N}\\ 
- 2\pi a(R_0) N^2(N-1)(1-\theta)  \int dx_2 \cdots\int dx_N \, \chi_{|\vx_3-\vx_2|\leq 2R} |\Psi(\vx_2,\vx_2,\vx_3,\dots,\vx_N)|^2.
\label{eq:last_lower_bound}
\end{multline}
By Carlen's inequality \eqref{Carlen_inequality} with $p=0$ and $\vx_2,\vx_4,\dots,\vx_N$ fixed, we have
$$|\Psi(\vx_2,\vx_2,\vx_3,\dots,\vx_N)|^2\leq C\int_{\R^3}|\Psi(\vx_2,\vx_2,\vx_3,\dots,\vx_N)|^2\, dx_3 \,,$$
therefore
$$
 \int_{\R^3} dx_3 \, \chi_{|\vx_3-\vx_2|\leq 2R} |\Psi(\vx_2,\vx_2,\vx_3,\dots,\vx_N)|^2 \leq C R^3 \int_{\R^3} d \vx_3 \,  |\Psi(\vx_2,\vx_2,\vx_3,\dots,\vx_N)|^2\,.
$$
We conclude that
$$
(A+B)\restriction_{\gH_N} \geq \frac{(1-\theta)\left( t_R - C R^3 N\right)a(R_0)}{a} E_N^{\rm LLL}(\omega,a).
$$
Our final inequality is therefore
\begin{multline}
E_N(\omega,a ) \geq  \min\Bigg\{\frac{(1-\delta)(1-\theta)\left( t_R - C R^3 N\right)a(R_0)}{a} E_N^{\rm LLL}(\omega,a)\, , \\ \theta - \omega -\frac {6\pi a N^2}{\delta(R^3 - R_0^3)}\Bigg\}.
\end{multline}

We now optimize constants. Recall from \eqref{estim_energy_LLL} that
$E^{\rm LLL}_N(\omega,a)\leq \kappa^{-1}aN$. We choose $R_0=a^{1/9}$, $R^3=2R_0^3$, $\theta=\omega+ 2 a^{1/3} N $ and $\delta=6\pi a^{1/3} N$ such that
$$
\theta - \omega -\frac {6\pi a N^2}{\delta(R^3 - R_0^3)}=  Na^{1/3} \,.
$$
Assuming that $\kappa \geq  a^{2/3}$ this expression is greater than $E^{\rm LLL}_N(\omega,a)$. 
Recalling (\ref{estim_scattering_lower_bound}) the final result is then
\begin{equation}
 E_N(\omega,a ) \geq E^{\rm LLL}_N(\omega,a)\left(1-\omega-\frac 1{4\pi}\int_{|\vx|\geq a^{-8/9}} W(\vx)dx - C\left[ a^{1/3} N +a^{1/9}\right]\right)\,.
\end{equation}

If $\kappa < a^{2/3}$, the choice $R_0=\kappa^{1/6}$, $R^3=2R_0^3$, $\theta=\omega + 2 Na \kappa^{-1}$ and $\delta = 6\pi N \kappa^{1/2}$ yields the desired bound. 
This completes the proof of Theorem \ref{thm:bounds}.\qed
\medskip

\subsection{Proof of Theorem~\ref{cor:CV}}\label{proof:cor_CV}
The proof is a simple consequence of the  bounds in the previous subsection, together with the right choice of parameters. If we choose $R_0$, $R$ and $\delta$ as above, but $\theta$ to be bigger than the previous choice by an amount $\theta'\geq 0$, we conclude that
\begin{multline}
H^N_{\omega,a}\geq \left(1- \theta'-  C \left[ \frac{ N a^{1/3}}r + a^{1/9} r \right] - \frac{1}{4\pi}\int_{|\vx|\geq r^{1/6} a^{-8/9}} W \right) \Pi^*\tilde{H}^N_{\omega,a}\Pi\\
+ \left( \theta' + E^{\rm LLL}_N(\omega,a) \right) Q 
\end{multline}
where $\Pi=\pi^{N/4}e^{\sum_{j=1}^N|\vx_j|^2/2}P$ denotes the projection
from $L^2(\R^{3N})$ onto $\cB_N$, and $r=\min\{1,\kappa a^{-2/3}\}$ as in the statement of Theorem~\ref{thm:bounds}. Let
$\Pi_1:=\pi^{N/4} e^{\sum_{j=1}^N|\vx_j|^2/2}P_N(\kappa)$ denote the projection from $L^2(\R^{3N})$ onto 
the ground eigenspace of the operator $\cL_N/N^2+4\pi
\kappa\Delta_N/N$, and
$\gamma_N(\kappa)=(\lambda_2-\lambda_1)/\lambda_1>0$ where $\lambda_1$
and $\lambda_2$ are the first and second eigenvalues of
this operator, respectively. Introducing $\Pi_2=\Pi-\Pi_1$, we have
$$\Pi^*\tilde{H}^N_{\omega,a}\Pi= \Pi^*_1\tilde{H}^N_{\omega,a}\Pi_1+\Pi^*_2\tilde{H}^N_{\omega,a}\Pi_2\geq E_N^{\rm LLL}(\omega,a)\,P_N(\kappa)+\Pi^*_2\tilde{H}^N_{\omega,a}\Pi_2.$$
Note that
\begin{align*}
\Pi^*_2\tilde{H}^N_{\omega,a}\Pi_2&=N^2\omega\Pi^*_2\left(\frac{\cL_N}{N^2}+4\pi\frac{a}{N\omega}\frac{\Delta_N}{N}\right)\Pi_2\\
&\geq N^2\omega\left(1-\left|\frac{a}{N\omega\kappa}-1\right|\right)\Pi^*_2\left(\frac{\cL_N}{N^2}+4\pi\kappa\frac{\Delta_N}{N}\right)\Pi_2\\
&\geq E^{\rm LLL}_N(\omega,a)(1+\gamma_N(\kappa))\left(1-\left|\frac{a}{N\omega\kappa}-1\right|\right)\left(1-\left|\frac{N\omega\kappa}{a}-1\right|\right)\Pi^*_2\Pi_2.
\end{align*}
For $a$, $\omega$, and $\kappa-a/(N\omega)$ small enough, we obtain
\begin{multline*}
H^N_{\omega,a}\geq E^{\rm LLL}_N(\omega,a) \left(1- \theta'-  C \left[ \frac{ N a^{1/3}}r + a^{1/9} r \right] - \frac{1}{4\pi}\int_{|\vx|\geq r^{1/6} a^{-8/9}} W \right)\\+E^{\rm LLL}_N(\omega,a)\frac{\gamma_N(\kappa)}{2}\Pi^*_2\Pi_2+ \theta' Q\,.
\end{multline*}
Combined with the upper bound \eqref{eq:upper_bound}, this clearly yields $\norm{\Pi_2\Psi^N_{\omega,a}}\to0$ and $\norm{Q\Psi^N_{\omega,a}}\to0$, as was claimed.\qed

\bigskip

\noindent\textbf{Acknowledgments.} The authors would like to thank Elliott H. Lieb for fruitful discussions. Financial support from the ANR project ACCQuaRel of the French ministry of research (M.L) and the U.S. National Science Foundation under grant No. PHY-0652356  (R.S.) is gratefully acknowledged. 

%%%%%%%%%%%%%%%%%%%%%%%%%%%%%%%%%%%%%%%%%%%%%
%%%%%%%%%%%%%%%%%%%%%%%%%%%%%%%%%%%%%%%%%%%%%
\bibliographystyle{siam}
%\bibliography{biblio}

\end{document}